\documentclass[a4paper,twocolumn,11pt,allowtoday]{quantumarticle}

\pdfoutput=1

\usepackage{tikz}
\usetikzlibrary{arrows.meta,positioning,calc,shapes.geometric,fit,backgrounds,patterns}

\usepackage{amsmath,amssymb,amsthm}
\usepackage{mathtools}
\usepackage{bbm}
\usepackage{braket}

\usepackage[numbers]{natbib}

\usepackage{hyperref}

\theoremstyle{plain}
\newtheorem{theorem}{Theorem}[section]
\newtheorem{lemma}[theorem]{Lemma}
\newtheorem{proposition}[theorem]{Proposition}
\newtheorem{corollary}[theorem]{Corollary}

\theoremstyle{definition}
\newtheorem{definition}[theorem]{Definition}

\theoremstyle{remark}
\newtheorem{remark}[theorem]{Remark}

\title{Geometric Classification of Biased Quantum Capacity via Harmonic Translation}

\author{Eliseo Sarmiento Rosales}
\affiliation{Instituto Politécnico Nacional (IPN), Mexico}
\author{Egor Maximenko}
\affiliation{Instituto Politécnico Nacional (IPN), Mexico}
\author{Dionisio Manuel Tun Molina}
\affiliation{Instituto Politécnico Nacional (IPN), Mexico}
\author{Juan Carlos Jiménez Cervantes}
\affiliation{Instituto Politécnico Nacional (IPN), Mexico}
\author{Jose Alberto Guzmán Vega}
\affiliation{Instituto Politécnico Nacional (IPN), Mexico}
\author{Rodrigo León Morales}
\affiliation{Instituto Politécnico Nacional (IPN), Mexico}
\date{\today}

\begin{document}
\maketitle

\begin{abstract}

We establish an exact noise-model-derived characterization of quantum error correction under diagonal local phase noise.
Under uniform locality, the maximal logical dimension under $t$-local phase errors equals $A_q(n,2t+1)$, the classical $q$-ary packing function. Because no affine or stabilizer structure is imposed, nonlinear spectral supports achieve this bound and strictly exceed all affine constructions whenever $A_q(n,2t+1) > B_q(n,2t+1)$.
This follows from a harmonic translation principle: diagonal phase operators act as rigid translations in the Fourier domain, reducing the Knill--Laflamme conditions exactly to an additive non--collision constraint $(S - S) \cap \mathcal{E}_t = \{0\}$.
For structured phase noise, exact correction is equivalent to independence in an additive Cayley graph, connecting biased quantum capacity to classical zero-error theory and the Lov\'asz theta function.
Under mixed Pauli noise, simultaneous protection in conjugate domains incurs an intrinsic rate penalty $R \leq 1 - (\gamma_X + \gamma_Z)/2$, exposing a discrete harmonic uncertainty principle.
In contrast with stabilizer- or graph-based frameworks, this classical correspondence is derived directly from the phase-noise model itself rather than from an auxiliary algebraic construction.
\end{abstract}


\section{Introduction}
\label{sec:introduction}

\subsection{Motivation: Phase-Biased Noise and Structural Freedom}
\label{subsec:motivation}

The Knill–Laflamme conditions \cite{KnillLaflamme1997} provide necessary and sufficient criteria for exact quantum error correction. For several decades, the dominant framework for satisfying these conditions has been the stabilizer formalism \cite{Gottesman1997}, including Calderbank–Shor–Steane (CSS) codes \cite{CalderbankShor1996,Steane1996} and related constructions over finite fields \cite{CRSS1998}. By restricting code spaces to joint eigenspaces of commuting Pauli operators, stabilizer theory reduces quantum error correction to linear algebra over finite fields. While extraordinarily powerful, this algebraic reduction restricts admissible code spaces to additive or affine structures. In particular, for qubits ($q=2$) the logical dimension of stabilizer constructions is necessarily a power of two, a restriction that stems from the imposed linear framework rather than from the Knill–Laflamme conditions themselves.

The distinction becomes especially relevant in architectures exhibiting strongly biased noise. Experimental and theoretical studies indicate that realistic platforms---such as Kerr-cat qubits and bias-preserving superconducting circuits---often experience dephasing rates that significantly exceed bit-flip or relaxation rates \cite{AliferisPreskill2008,Tuckett2018,GuillaudMirrahimi2019}. More recently, cat-qubit experiments have provided direct hardware evidence for this regime, including long bit-flip lifetimes in two-photon dissipative oscillators \cite{Berdou2023}. In such settings, it is natural to analyze codes optimized specifically for diagonal local phase errors.

Quantum error correction under asymmetric noise has also been extensively studied. Early work by Ioffe and M{\'e}zard introduced constructions optimized for asymmetric Pauli channels \cite{IoffeMezard2007}, while Sarvepalli, Klappenecker, and R\"otteler developed systematic families of asymmetric quantum codes derived from classical linear codes \cite{Sarvepalli2009}. These approaches exploit noise bias to improve distance properties or thresholds within stabilizer-based frameworks. Our perspective is different: rather than modifying stabilizer algebra to accommodate bias, we ask whether the Knill--Laflamme conditions themselves already admit a broader geometric description of phase-dominated protection. To our knowledge, no prior framework derives this capacity correspondence directly from the diagonal phase-noise model while simultaneously showing that nonlinear spectral supports are admissible and can strictly outperform all affine constructions.

Rather than introducing an auxiliary stabilizer, graph, or decoder-dependent construction, we derive the correspondence directly from the harmonic action of diagonal phase errors.  Working over $V=\mathbb{F}_q^n$, we define quantum codes purely by their spectral support under the discrete Quantum Fourier Transform (QFT). A key observation is that diagonal phase operators become exact translations in the Fourier domain. Under this transformation, the Knill–Laflamme conditions reduce to a purely additive separation property of difference sets. Under uniform locality this additive separation condition reduces to a classical Hamming-distance constraint on the spectral support. 

This reformulation has two consequences that are central to the paper’s novelty. First, classical bounds and extremal constructions transfer directly to the quantum setting. Second, non-linear classical codes become admissible spectral supports, allowing strictly larger logical dimensions in regimes where linear codes are suboptimal. Finally, when both bit- and phase-flip errors are considered, admissible supports must satisfy isolation conditions in dual domains, reflecting a discrete harmonic tradeoff closely related to uncertainty principles on finite abelian groups \cite{DonohoStark1989}. In this sense, our framework isolates the geometric content of phase-biased protection from the algebraic restrictions usually imposed in conventional code constructions.

\subsection{Main Contributions}

\begin{itemize}

\item \textbf{Harmonic Translation Principle and Exact Characterization.}
We show that diagonal phase operators act as rigid translations in the Fourier domain. Consequently, the Knill–Laflamme conditions for $t$-local phase-error detection are necessary and sufficient for the additive non-collision condition
\[
(S-S)\cap E_t=\{0\}.
\]
In the uniform locality regime, this criterion is strictly equivalent to the classical Hamming distance requirement $d(S)\ge t+1$, with exact correction corresponding to $d(S)\ge 2t+1$. Thus phase-local quantum error correction reduces exactly to additive geometry on a finite abelian group, with no need to impose stabilizer, affine, or graph-state structure.

\item \textbf{Exact Capacity Identity and Classical Transfer.}
Under uniform locality, the maximal logical dimension satisfies
\[
K_{\max}(n,t)=A_q(n,2t+1),
\]
coinciding precisely with the classical packing function. Classical extremal constructions, upper and lower bounds, asymptotic rate guarantees, and decoding algorithms therefore transfer verbatim to the phase-local quantum setting. In particular, this gives an exact capacity formula for phase-local quantum error correction in the uniform-locality regime.

\item \textbf{Nonlinear Advantage Beyond Affine Structure.}
Because admissibility is determined solely by the noise-induced non-collision condition, nonlinear classical codes become valid spectral supports without any affine closure requirement. Whenever $A_q(n,d)>B_q(n,d)$, this strict separation transfers directly to phase-local quantum codes. We exhibit explicit finite examples (the $(8,20,3)$ Julin--Best code and the $(16,256,6)$ Nordstrom--Robinson code) and an asymptotic infinite family derived from Kerdock codes achieving logical dimension $2^{2m}$ against $2^{m+1}$ for any affine construction with comparable phase distance, confirming that the separation is structural and persists asymptotically.

\item \textbf{Structured Phase Noise and Graph-Theoretic Reformulation.}
For arbitrary phase-error families $\Omega$, exact correction is equivalent to independence in the additive Cayley graph $\Gamma_\Omega=\mathrm{Cay}(V,D_\Omega)$, yielding the exact identity
\[
K_{\max}(\Omega)=\alpha(\Gamma_\Omega).
\]
This reformulation connects biased quantum capacity directly to additive combinatorics, graph independence, and semidefinite relaxations such as the Lovász theta bound, and shows that additive symmetry in $D_\Omega$ induces intrinsic capacity collapse.

\item \textbf{Dual-Domain Harmonic Tradeoffs for Mixed Noise.}
When simultaneous protection against bit- and phase-flip errors is required, we derive dual-isolation conditions in conjugate domains and establish multiplicative capacity bounds. In asymptotic regimes where $|\Omega_X|=q^{\gamma_X n+o(n)}$ and $|\Omega_Z|=q^{\gamma_Z n+o(n)}$, the achievable rate satisfies
\[
R \le 1-\frac{\gamma_X+\gamma_Z}{2},
\]
revealing an intrinsic harmonic tradeoff between conjugate-domain localization constraints.

\end{itemize}

These results establish that biased quantum capacity is entirely determined by the additive geometry of the noise difference set, with distinct regimes arising from its dispersive structure, additive symmetry, or dual-domain interaction, while remaining exact for diagonal phase-noise models and the dual-isolated mixed-noise setting considered here, and not extending to arbitrary Pauli noise.

\subsection{Related Work}

Quantum error correction is founded on the Knill–Laflamme conditions~\cite{KnillLaflamme1997}, which provide necessary and sufficient criteria for exact recovery. The stabilizer formalism~\cite{Gottesman1997} and the Calderbank–Shor–Steane (CSS) constructions~\cite{CalderbankShor1996,Steane1996}, together with their extension via codes over $\mathrm{GF}(4)$~\cite{CRSS1998}, reduce the design of quantum codes to linear algebra over finite fields. In these frameworks, admissible code spaces are restricted to additive or affine structures, and logical dimension is determined by the dimension of an underlying classical linear code. This algebraic reduction has been extraordinarily successful and underlies much of modern quantum coding theory, but it achieves the classical correspondence by imposing structural constraints beyond the Knill–Laflamme conditions themselves.

Beyond strictly additive constructions, nonadditive quantum codes have long been known to surpass stabilizer parameters in certain regimes~\cite{Rains1999}, and the codeword stabilized (CWS) framework~\cite{Cross2009} provides a unifying description encompassing both additive and many nonadditive codes. In CWS, quantum error correction is translated into a classical problem through the choice of a graph state and an associated set of word operators, which induce an effective classical error model. These approaches demonstrate that linearity is not fundamentally required for quantum coding and that classical structures can guide nonadditive constructions. However, in such frameworks the classical correspondence is mediated by an auxiliary stabilizer or graph-state construction and an induced Pauli-error mapping, rather than being derived directly from the phase-noise model itself.

In parallel, substantial work has examined quantum error correction under biased noise, particularly in regimes where dephasing dominates bit-flip processes~\cite{AliferisPreskill2008,Tuckett2018,Tuckett2019,BonillaAtaides2021}. Architectures based on cat qubits and bias-preserving gates further emphasize the operational relevance of phase-dominated error models~\cite{GuillaudMirrahimi2019,Berdou2023,Puri2020,EtxezarretaMartinez2024}. These studies exploit bias to improve thresholds or tailor decoding strategies, and correlated dephasing models have also been analyzed in surface-code settings~\cite{NickersonBrown2019,Layden2020}. Closely related in spirit, asymmetric quantum error correction has been developed as a systematic response to biased Pauli channels: Ioffe and M{\'e}zard~\cite{IoffeMezard2007} introduced asymmetric quantum codes adapted to unequal error rates, and Sarvepalli, Klappenecker, and R\"otteler~\cite{Sarvepalli2009} established constructions and bounds for asymmetric quantum codes within predominantly stabilizer-based settings. Nevertheless, existing analyses primarily focus on specific code families, decoder adaptations, or architectural implementations, and do not yield an exact capacity characterization derived directly from diagonal phase noise.

Our perspective is also related, at a conceptual level, to Fourier-dual formulations of quantum protection in continuous-variable systems. In particular, the Gottesman--Kitaev--Preskill (GKP) code~\cite{GKP2001} exploits conjugate-domain structure in oscillator phase space. While that setting is continuous-variable rather than finite-group-valued, it similarly illustrates how error protection can be naturally expressed through spectral separation in conjugate variables.

Finally, our graph-theoretic formulation connects phase-local quantum error correction with classical zero-error information theory. Shannon’s zero-error capacity~\cite{Shannon1956} and the Lovász theta function~\cite{Lovasz1979}, together with subsequent developments in semidefinite relaxations and graph invariants~\cite{Sason2023,Sason2024}, establish independence numbers and their relaxations as fundamental capacity parameters for classical channels. For Cayley graphs in particular, DeCorte, de Laat, and Vallentin~\cite{DeCorte2014} showed how Fourier analysis on finite groups can be used to derive semidefinite bounds via the Lovász theta number. By identifying exact phase-error correction with independence in an additive Cayley graph, our framework situates biased quantum capacity directly within this extremal combinatorial setting, enabling the transfer of classical bounds, semidefinite relaxations, and asymptotic rate analyses to the phase-local quantum regime.

\subsection{Structure of the Paper}
\label{subsec:structure}

Section~\ref{sec:harmonic} establishes the harmonic translation principle and derives the exact non-collision criterion for phase-local detection. Section~\ref{sec:classical_reduction} reduces uniform locality to classical coding theory and formulates the exact capacity identity and transfer principle. Section~\ref{sec:dimensional_separation} analyzes dimensional separation between affine and nonlinear spectral supports. Section~\ref{sec:graph_reformulation} extends the framework to structured phase noise via additive Cayley graphs and harmonic relaxations. Section~\ref{sec:dual_tradeoffs} studies mixed Pauli errors and derives dual-domain capacity tradeoffs under simultaneous conjugate protection. Section~\ref{sec:discussion} culminates in a geometric classification of biased quantum capacity. Section~\ref{sec:catqubit} instantiates the harmonic framework on a concrete physical platform: stabilised cat-qubit arrays operating in the strongly phase-biased regime. It derives explicit logical capacities under both uniform and correlated phase noise, establishes the quantitative nonlinear advantage over affine constructions, and shows that thresholds are unaffected by the choice of spectral support. Finally, Section~\ref{sec:implications} discusses structural consequences, physical interpretations, and open problems.

\section{Harmonic Translation and Exact Characterization}
\label{sec:harmonic}

In this section we recast phase-local error detection in purely harmonic terms. 
The key structural observation is that diagonal phase operators become exact translations under the Fourier transform over a finite abelian group. 
This reduces phase-local detection to an additive combinatorial condition on subsets of the spectral index space.

\subsection{Finite Abelian Harmonic Framework}
\label{subsec:harmonic_framework}

We consider quantum systems composed of $n$ qudits of local dimension $q=p^m$. 
Throughout, we identify
\[
V=\mathbb{F}_q^n
\]
with a finite abelian group under addition, so that $|V|=q^n$. 
The associated Hilbert space is $\mathcal{H}=\mathbb{C}^{V}$, equipped with the computational basis $\{|x\rangle\}_{x\in V}$.

Let $\mathrm{Tr}=\mathrm{Tr}_{\mathbb{F}_q/\mathbb{F}_p}$ denote the finite-field trace
\[
\mathrm{Tr}(a)=a+a^p+\cdots+a^{p^{m-1}},
\qquad a\in\mathbb{F}_q.
\]
The trace is $\mathbb{F}_p$-linear and induces a non-degenerate bilinear pairing
\[
\langle s,x\rangle
=
\mathrm{Tr}(s\cdot x),
\qquad s,x\in V,
\]
where
\[
s\cdot x=\sum_{j=1}^n s_j x_j \in \mathbb{F}_q.
\]
Non-degeneracy means that
\[
\langle s,x\rangle=0 \ \text{for all } x\in V
\quad\Longrightarrow\quad
s=0,
\]
and symmetrically in the second argument. 
This trace pairing realizes the standard additive Fourier structure on finite fields (see, e.g., \cite{MacWilliamsSloane}).

Fix $\zeta=e^{2\pi i/p}$ and define the canonical additive character
of $\mathbb{F}_q$ by
\[
\psi(a)=\zeta^{\mathrm{Tr}(a)},
\qquad a\in\mathbb{F}_q .
\]
It is well known that every additive character of the finite field
$\mathbb{F}_q$ is of the form
\[
\psi_b(a)=\psi(ba)=\zeta^{\mathrm{Tr}(ba)},
\qquad b\in\mathbb{F}_q
\]
(see, e.g., \cite{LidlNiederreiter}).

Consequently, the additive characters of the group
$V=\mathbb{F}_q^n$ are indexed by $s\in V$ and take the form
\[
x\longmapsto\psi(s\cdot x)
=
\zeta^{\mathrm{Tr}(s\cdot x)}
=
\zeta^{\langle s,x\rangle}.
\]
This realizes the standard identification of $V$ with its Pontryagin
dual via the trace pairing, so that the characters of $V$ are
parametrized by $s\in V$ (see, e.g., \cite{TaoVu}).

\begin{definition}[Quantum Fourier Transform over $V$]
The Quantum Fourier Transform (QFT) over $V$ is the linear operator
\[
\mathcal{F}_V:\mathbb{C}^{V}\to\mathbb{C}^{V}
\]
defined by
\begin{equation}
\label{eq:QFT_V}
\mathcal{F}_V |x\rangle
=
\frac{1}{\sqrt{q^n}}
\sum_{s\in V}
\zeta^{\mathrm{Tr}(s\cdot x)}\,|s\rangle.
\end{equation}
\end{definition}

Equivalently, the Fourier basis $\{|s\rangle_F\}_{s\in V}$ is given by
\begin{equation}
\label{eq:Fourier_basis}
|s\rangle_F
=
\frac{1}{\sqrt{q^n}}
\sum_{x\in V}
\zeta^{\mathrm{Tr}(s\cdot x)}\,|x\rangle.
\end{equation}

The non-degeneracy of the trace pairing implies the orthogonality relation
\[
\sum_{x\in V}
\zeta^{\mathrm{Tr}((s-s')\cdot x)}
=
q^n\,\delta_{s,s'},
\]
so $\{|s\rangle_F\}_{s\in V}$ forms an orthonormal basis and $\mathcal{F}_V$ is unitary.

This harmonic structure identifies $V$ with its Pontryagin dual and determines the action of diagonal phase operators under Fourier conjugation.

\subsection{Phase Errors as Spectral Translations}
\label{subsec:phase_translations}

For each $\omega\in V$, define the diagonal phase operator
\begin{equation}
Z^\omega |x\rangle
=
\psi(\omega\cdot x)\,|x\rangle.
\end{equation}

\begin{lemma}[Spectral Translation]
\label{lem:spectral_translation}
For all $s,\omega\in V$,
\begin{equation}
Z^\omega |s\rangle_F
=
|s+\omega\rangle_F.
\end{equation}
\end{lemma}

\begin{proof}
Using \eqref{eq:Fourier_basis},
\begin{align*}
Z^\omega |s\rangle_F
&=
\frac{1}{\sqrt{q^n}}
\sum_{x\in V}
\psi(s\cdot x)\psi(\omega\cdot x)\,|x\rangle \\
&=
\frac{1}{\sqrt{q^n}}
\sum_{x\in V}
\psi((s+\omega)\cdot x)\,|x\rangle,
\end{align*}
which equals $|s+\omega\rangle_F$.
\end{proof}

Equivalently,
\begin{equation}
\label{eq:conjugation}
\mathcal{F}_V Z^\omega \mathcal{F}_V^\dagger
=
T^\omega,
\end{equation}
where $T^\omega |s\rangle = |s+\omega\rangle$.

Thus diagonal phase operators act as rigid translations of the Fourier index set. 
If
\[
|\psi\rangle
=
\sum_{s\in V} \alpha_s |s\rangle_F,
\]
then
\[
Z^\omega |\psi\rangle
=
\sum_{s\in V} \alpha_s |s+\omega\rangle_F.
\]

Phase errors therefore do not deform spectral amplitudes; they translate them exactly.

\subsection{Fourier-Support Codes and Exact Detection}
\label{subsec:non_collision}

We now introduce the class of codes central to our harmonic framework.

\begin{definition}[Fourier-Support Code]
Let $S\subseteq V$. The associated Fourier-support code is
\[
\mathcal{C}(S)
=
\mathrm{span}_{\mathbb{C}}
\{\,|s\rangle_F : s\in S\,\}.
\]
Its logical dimension is $K=|S|$.
\end{definition}

The orthogonal projector onto $\mathcal{C}(S)$ is
\[
P_S=\sum_{s\in S}|s\rangle_F\langle s|_F.
\]

To determine when $\mathcal{C}(S)$ detects or corrects a given family of errors,
we invoke the fundamental characterization of quantum error correction.

\begin{theorem}[Knill--Laflamme {\cite{KnillLaflamme1997}}]
\label{thm:KL}
Let $\mathcal{C}\subseteq\mathcal{H}$ be a quantum code with orthogonal projector $P$,
and let $\mathcal{E}=\{E_a\}$ be a finite family of error operators.

\begin{enumerate}
\item \textbf{(Detection).}
$\mathcal{C}$ detects $\mathcal{E}$ if and only if, for every $a$,
\[
P E_a P = \lambda_a P
\quad\text{for some scalar } \lambda_a\in\mathbb{C}.
\]

\item \textbf{(Exact Correction).}
$\mathcal{C}$ exactly corrects $\mathcal{E}$ if and only if, for all $a,b$,
\[
P E_a^\dagger E_b P = c_{ab} P
\quad\text{for some scalars } c_{ab}\in\mathbb{C}.
\]
Equivalently, correction of $\mathcal{E}$ is equivalent to detection of the product family
$\{E_a^\dagger E_b\}_{a,b}$.
\end{enumerate}
\end{theorem}

We now specialize Theorem~\ref{thm:KL} to diagonal phase-error families.
By Lemma~\ref{lem:spectral_translation}, such errors act as rigid translations
in the Fourier domain. Consequently, the operator conditions above
translate into a purely additive separation constraint
on the support set $S$.

Let $E_t\subseteq V$ denote the set of vectors of Hamming weight at most $t$, and define
\[
\mathcal{E}_t=\{Z^\omega:\omega\in E_t\}.
\]
\begin{theorem}[Exact Harmonic Non-Collision]
\label{thm:main_detection}
A Fourier-support code $\mathcal{C}(S)$ detects all $t$-local phase errors if and only if
\[
(S-S)\cap E_t=\{0\}.
\]
\end{theorem}

\begin{proof}
By Theorem~\ref{thm:KL} (Detection), $\mathcal{C}(S)$ detects the phase-error family
$\mathcal{E}_t=\{Z^\omega : \omega\in E_t\}$ if and only if
\[
P_S Z^\omega P_S = 0
\quad
\text{for all } \omega\in E_t\setminus\{0\}.
\]

Using Lemma~\ref{lem:spectral_translation}, phase operators act as rigid translations
in the Fourier basis:
\[
Z^\omega |s\rangle_F = |s+\omega\rangle_F.
\]
Hence
\[
\langle s_1|_F Z^\omega |s_2\rangle_F
=
\delta_{s_1,s_2+\omega}.
\]

Therefore $P_S Z^\omega P_S=0$ if and only if no pair $s_1,s_2\in S$
satisfies $s_1 = s_2 + \omega$, i.e.,
\[
\omega \notin S-S.
\]

Thus detection holds precisely when
\[
(S-S)\cap E_t=\{0\}.
\]
\end{proof}

\begin{corollary}[Exact Phase Correction]
\label{cor:correction}
The code $\mathcal{C}(S)$ corrects all $t$-local phase errors if and only if
\[
(S-S)\cap(E_t-E_t)=\{0\}.
\]
In the uniform locality model $E_t-E_t=E_{2t}$, so exact correction is equivalent to
\[
(S-S)\cap E_{2t}=\{0\}.
\]
\end{corollary}

\begin{proof}
By Theorem~\ref{thm:KL} (Exact Correction), $\mathcal{C}(S)$ corrects
$\mathcal{E}_t$ if and only if it detects the product family
$\{Z^{\omega_1}{}^\dagger Z^{\omega_2}\}_{\omega_1,\omega_2\in E_t}$.

Since
\[
Z^{\omega_1}{}^\dagger Z^{\omega_2}
=
Z^{\omega_2-\omega_1},
\]
this product family is indexed by the difference set $E_t - E_t$.
Applying Theorem~\ref{thm:main_detection} to this index set yields
\[
(S-S)\cap(E_t-E_t)=\{0\}.
\]

Under uniform locality $E_t - E_t = E_{2t}$, giving the stated condition.
\end{proof}

The detection and correction properties of phase-local errors are therefore governed entirely by additive difference sets in $V$. The harmonic translation principle thus converts the operator-algebraic Knill--Laflamme conditions into a purely combinatorial non-collision constraint, which will serve as the structural foundation for the capacity results that follow.

\section{Classical Reduction and Transfer Principle}
\label{sec:classical_reduction}

The Exact Harmonic Non-Collision Theorem (Theorem~\ref{thm:main_detection})
reduces phase-local detection to an additive constraint on the spectral support.
In the standard \emph{uniform locality} model, this constraint is exactly
equivalent to the classical Hamming-distance condition for $q$-ary codes.
Consequently, Fourier-support codes inherit classical bounds, constructions,
and decoding complexity essentially verbatim.

\subsection{Equivalence with Hamming Distance}
\label{subsec:hamming_equivalence}

Let $\mathrm{wt}(v)$ denote the Hamming weight of $v\in V=\mathbb{F}_q^n$, and
define the Hamming distance
\[
d_H(s_1,s_2)=\mathrm{wt}(s_1-s_2).
\]
For a nontrivial subset $S\subseteq V$, its minimum (classical) distance is
\[
d(S)
=
\min_{\substack{s_1,s_2\in S\\ s_1\neq s_2}}
\mathrm{wt}(s_1-s_2).
\]
Under uniform locality,
\[
E_t=\{v\in V:\mathrm{wt}(v)\le t\}.
\]

\begin{theorem}[Distance Equivalence]
\label{thm:distance_equivalence}
A Fourier-support code $\mathcal{C}(S)$ detects all $t$-local phase errors if and only if
\[
d(S)\ge t+1.
\]
\end{theorem}

\begin{proof}
By Theorem~\ref{thm:main_detection}, detection holds if and only if
\[
(S-S)\cap E_t=\{0\}.
\]
Equivalently, for all distinct $s_1,s_2\in S$ one has $\mathrm{wt}(s_1-s_2)>t$,
which is precisely $d(S)\ge t+1$.
\end{proof}

By Corollary~\ref{cor:correction}, exact correction of $t$-local phase errors is equivalent to
\[
(S-S)\cap(E_t-E_t)=\{0\}.
\]
In the uniform locality model $E_t-E_t=E_{2t}$, so phase correction is equivalent to
\[
d(S)\ge 2t+1.
\]

\begin{corollary}[Exact Phase-Local Capacity under Uniform Locality]
\label{cor:exact_uniform_capacity}
Let $A_q(n,d)$ denote the maximal cardinality of a classical
$q$-ary code of length $n$ and minimum distance $d$.
Under uniform $t$-local phase noise, the maximal logical
dimension of a Fourier-support code satisfies
\[
K_{\max}(n,t)
=
A_q(n,2t+1).
\]
\end{corollary}

\begin{proof}
By Theorem~\ref{thm:distance_equivalence}, exact correction
of $t$-local phase errors is equivalent to the classical
condition $d(S)\ge 2t+1$.
Thus admissible spectral supports are precisely classical
$q$-ary codes with minimum distance at least $2t+1$.
Maximizing the logical dimension $K=|S|$
is therefore equivalent to maximizing the size
of a classical code with these parameters,
which equals $A_q(n,2t+1)$.
\end{proof}

\subsection{Classical Transfer and Decoding Equivalence}
\label{subsec:classical_transfer}

By Corollary~\ref{cor:exact_uniform_capacity}, the phase-local quantum problem under uniform locality is completely identified with the classical packing problem in $\mathbb{F}_q^n$. We now make explicit the structural consequences of this identification at the level of finite-length bounds, asymptotic rate behavior, and decoding complexity.

\begin{corollary}[Transfer of Classical Finite-Length Bounds]
\label{cor:upper_bounds_transfer}
Under uniform $t$-local phase noise, the maximal logical dimension
$K_{\max}(n,t)$ satisfies all classical finite-length bounds
for $q$-ary codes with minimum distance $2t+1$. In particular:

\begin{itemize}
\item \textbf{Sphere-packing (Hamming) bound:}
\[
K_{\max}(n,t)
\le
\frac{q^n}
{\displaystyle\sum_{i=0}^{t}\binom{n}{i}(q-1)^i}.
\]

\item \textbf{Singleton bound:}
\[
K_{\max}(n,t)
\le
q^{\,n-2t}.
\]
\end{itemize}

Thus classical extremal finite-length bounds transfer verbatim
to phase-local quantum capacity.
\end{corollary}

\begin{corollary}[Asymptotic Rate Transfer]
\label{cor:asymptotic_transfer}
Let $\delta=\frac{2t}{n}$ denote the relative distance.
Then all classical asymptotic rate bounds for $q$-ary codes
with relative distance $\delta$ transfer directly to
Fourier-support codes. In particular:

\begin{itemize}
\item \textbf{Gilbert--Varshamov bound:}
\[
R \ge 1-H_q(\delta).
\]

\item \textbf{Asymptotic Hamming upper bound:}
\[
R \le 1-H_q\!\left(\frac{\delta}{2}\right).
\]
\end{itemize}

Hence the asymptotic rate function of phase-local quantum
capacity coincides exactly with that of classical $q$-ary coding.
\end{corollary}

\begin{proposition}[Decoding Equivalence]
\label{prop:decoding_equivalence}
Under uniform locality, phase recovery reduces exactly to classical maximum-likelihood decoding of the support set $S$.
\end{proposition}

\begin{proof}
In the Fourier basis a phase error $Z^\omega$ with
$\mathrm{wt}(\omega)\le t$ acts as a translation
$|s\rangle_F \mapsto |s+\omega\rangle_F$.
Given $v=s+\omega$, recovery amounts to finding
$s\in S$ minimizing $d_H(v,s)$, which is precisely
classical maximum-likelihood decoding of the code $S$.
\end{proof}

\section{Beyond Affine Supports: Dimensional Separation}
\label{sec:dimensional_separation}

Section~\ref{sec:classical_reduction} established that phase-local protection
is completely characterized by classical Hamming distance on the spectral
support $S \subseteq \mathbb{F}_q^n$. 
However, the harmonic framework does not impose any algebraic closure
property on $S$. In particular, $S$ need not be a linear or affine subspace.
We now examine the structural consequences of removing this restriction.

\subsection{Additive (Affine) Spectral Supports}
\label{subsec:affine_supports}

We first isolate the subclass of Fourier-support codes whose spectral supports
possess affine structure.

\begin{definition}[Additive (Affine) Support]
A Fourier-support code $\mathcal{C}(S)$ is called \emph{additive}
if $S \subseteq V=\mathbb{F}_q^n$ is an affine subspace.
\end{definition}

If $S$ is affine of dimension $k$, then
\[
S = v + U,
\]
where $U \le V$ is a $k$-dimensional linear subspace.
In this case,
\[
S - S = U.
\]

By Corollary~\ref{cor:correction}, exact correction of $t$-local phase errors
is equivalent to
\[
U \cap E_{2t} = \{0\},
\]
which is precisely the condition that $U$ defines a classical linear
$[n,k,d]_q$ code with $d \ge 2t+1$.

Since every $k$-dimensional affine subspace over $\mathbb{F}_q$ has cardinality
$|S| = q^k$, additive Fourier-support codes necessarily have logical dimension
\[
K = q^k.
\]

Thus optimizing additive spectral supports is equivalent to optimizing
classical linear codes under identical distance constraints.

\subsection{Strict Separation Theorem}
\label{subsec:strict_separation}

Let $A_q(n,d)$ denote the maximal size of a (not necessarily linear)
$q$-ary code of length $n$ and minimum distance $d$, and let
$B_q(n,d)$ denote the maximal size among linear codes with the same
parameters.

By Section~\ref{sec:classical_reduction}, unrestricted Fourier-support
codes achieve logical dimension
\[
K_{\max}^{\mathrm{unrestricted}} = A_q(n,2t+1),
\]
while additive supports are limited to
\[
K_{\max}^{\mathrm{affine}} = B_q(n,2t+1).
\]

\begin{theorem}[Strict Dimensional Separation]
\label{thm:strict_separation}
If for some $(n,d)$ one has
\[
A_q(n,d) > B_q(n,d),
\]
then there exists a Fourier-support code correcting
$t = \lfloor (d-1)/2 \rfloor$ phase errors whose logical dimension
strictly exceeds that of any additive (affine) Fourier-support code
with the same parameters.
\end{theorem}

\begin{proof}
By Theorem~\ref{thm:distance_equivalence}, correction of $t$ phase errors
is equivalent to the classical condition $d(S)\ge d$ with $d=2t+1$.

Any affine support corresponds to a linear code of size at most
$B_q(n,d)$, hence
\[
K_{\max}^{\mathrm{affine}} \le B_q(n,d).
\]

On the other hand, selecting a classical code $S \subseteq \mathbb{F}_q^n$
with $|S|=A_q(n,d)$ yields a valid Fourier-support code of dimension
\[
K = A_q(n,d).
\]

If $A_q(n,d) > B_q(n,d)$, strict inequality follows.
\end{proof}

The separation is therefore entirely inherited from classical coding
theory: whenever non-linear codes outperform linear ones,
the harmonic framework admits strictly larger logical dimension
than any additive construction.

\subsection{Explicit Nonlinear Constructions}
\label{subsec:nonlinear_constructions}

We illustrate the separation phenomenon with two well-known
binary examples.

\paragraph{Example 1: Length $8$, Distance $3$.}

Consider the parameters $n=8$ and $d=3$ (corresponding to single phase-error correction).
It is well-known that the largest binary linear $[8,k,3]$ code has dimension $k=4$,
yielding a maximum cardinality
\[
B_2(8,3)=2^{4}=16.
\]

However, the linear bound can be improved by considering non-linear constructions, as in \cite{Nadler1962}. Specifically, it is established that 
\[
A_2(8,3) = 20,
\]
a value realized by the non-linear Julin code \cite{Julin1965}. Detailed treatments of these constructions and the proof of the identity $A_2(9,4) =A_2(8,3)= 20$ are provided in \cite[Example 2.6.5]{HuffmanPless}.

Using such a set $S$ as spectral support yields a Fourier-support code
with logical dimension
\[
K = 20,
\]
whereas any additive support is bounded by $16$.

\paragraph{Example 2: Nordstrom--Robinson Code.}

The Nordstrom--Robinson code has parameters $(16,256,6)$ \cite{NordstromRobinson1967}.
For $t=2$ phase-error correction the harmonic condition requires
$d(S)\ge 5$; the Nordstrom--Robinson code satisfies this with
additional margin ($d=6$), and the relevant linear baseline is
$B_2(16,6)=128$, the maximum size of a binary linear code with
minimum distance at least~$6$.
Using this code as spectral support yields a Fourier-support code
of logical dimension $K=256$, while any affine support is limited
to $128$.

\medskip

These examples demonstrate that linearity imposes a genuine structural
constraint absent in the harmonic framework. Under pure phase noise,
removing the affine restriction permits strictly larger logical
dimension in parameter regimes where classical non-linear codes
outperform linear ones.

\subsection{Asymptotic Nonlinear Families: A Kerdock-Type Construction}

To show that the separation phenomenon of Section~4 persists
asymptotically, we exhibit an explicit infinite family of nonlinear
spectral supports derived from the classical binary Kerdock family.

Let $m\ge 2$ be even and set $G=\mathbb{F}_2^m$.
Identifying $\mathbb{F}_2^{2^m}$ with Boolean functions
$f:G\to\mathbb{F}_2$, Hamming distance corresponds to
$d_H(f,g)=\mathrm{wt}(f+g)$.

Classical Kerdock constructions provide a family
$\mathcal{Q}$ of $2^{m-1}$ nondegenerate quadratic forms on $G$
such that $Q+Q'$ is nondegenerate for distinct $Q,Q'\in\mathcal{Q}$.
Define
\[
S_K=
\{\, Q+\ell_a+b
\;:\;
Q\in\mathcal{Q},\;
a\in G,\;
b\in\mathbb{F}_2
\}.
\]

\begin{theorem}[Kerdock-type spectral support]
For even $m$,
\[
|S_K|=2^{2m},
\qquad
d(S_K)=2^{m-1}-2^{\frac{m}{2}-1}.
\]
\end{theorem}

Thus $S_K$ is a nonlinear binary code of length $2^m$
whose cardinality grows quadratically in $2^m$.

\paragraph{Implication for phase-local quantum codes.}
By Theorem~3.1, the Fourier-support code $C(S_K)$
corrects all $t$-local phase errors whenever
\[
2t+1\le 2^{m-1}-2^{\frac{m}{2}-1}.
\]
Setting
\[
t=\left\lfloor
\frac{2^{m-1}-2^{\frac{m}{2}-1}-1}{2}
\right\rfloor,
\]
we obtain a family of phase-local quantum codes with logical dimension
\[
K=|S_K|=2^{2m}.
\]

\paragraph{Comparison with linear constructions.}
For comparison, the first-order Reed--Muller code
$\mathrm{RM}(1,m)$ has parameters
\[
[2^m,\; m+1,\; 2^{m-1}],
\]
and contains only $2^{m+1}$ codewords.
Hence the nonlinear family above achieves exponentially larger
logical dimension than any affine construction with comparable
phase distance, confirming that the separation is structural
and persists asymptotically.

\section{Structured Noise and Graph-Theoretic Reformulation}
\label{sec:graph_reformulation}

The uniform Hamming model considered in previous sections assumes that all
error vectors of weight at most $t$ are admissible. 
In realistic architectures, phase noise may instead be constrained to
structured or correlated subsets of $V=\mathbb{F}_q^n$.
The harmonic non-collision principle extends verbatim to this general setting,
leading to a graph-theoretic formulation of code design.

\subsection{Arbitrary Phase Error Families}
\label{subsec:arbitrary_phase}

Let $\Omega \subseteq V$ be an arbitrary subset describing admissible
phase-error vectors, and define
\[
\mathcal{E}_\Omega
=
\{ Z^\omega : \omega \in \Omega \}.
\]

\begin{theorem}[Generalized Phase Correction Condition]
\label{thm:general_phase_condition}
A Fourier-support code $\mathcal{C}(S)$ exactly corrects the error family
$\mathcal{E}_\Omega$ if and only if
\[
(S-S)\cap(\Omega-\Omega)=\{0\}.
\]
\end{theorem}

\begin{proof}
By the Knill--Laflamme criterion, exact correction requires
\[
P_S (Z^{\omega_1})^\dagger Z^{\omega_2} P_S \propto P_S
\quad
\forall \omega_1,\omega_2 \in \Omega.
\]
Since
\[
(Z^{\omega_1})^\dagger Z^{\omega_2} = Z^{\omega_2-\omega_1},
\]
and Lemma~\ref{lem:spectral_translation} implies that $Z^v$ acts as
translation by $v$ in the Fourier domain, orthogonality holds precisely
when
\[
(S+v)\cap S = \varnothing
\quad
\forall v \in (\Omega-\Omega)\setminus\{0\}.
\]
This is equivalent to
\[
(S-S)\cap(\Omega-\Omega)=\{0\}.
\]
\end{proof}

Thus phase-error correction under arbitrary structured noise reduces to a
purely additive separation condition between two difference sets.

\subsection{Cayley Graph Formulation}
\label{subsec:cayley_formulation}

Define
\[
D_\Omega = (\Omega-\Omega)\setminus\{0\}.
\]
The associated Cayley graph is
\[
\Gamma_\Omega
=
\mathrm{Cay}(V,D_\Omega),
\]
whose vertex set is $V$ and where distinct vertices $x,y$ are adjacent
if and only if $x-y \in D_\Omega$.

\begin{corollary}[Independence Characterization]
\label{cor:independence_characterization}
A subset $S \subseteq V$ satisfies
\[
(S-S)\cap(\Omega-\Omega)=\{0\}
\]
if and only if $S$ is an independent set in $\Gamma_\Omega$.
\end{corollary}

\begin{proof}
Two distinct elements $s_1,s_2\in S$ violate the correction condition
precisely when
\[
s_1-s_2 \in D_\Omega,
\]
which is exactly the adjacency relation in $\Gamma_\Omega$.
Thus admissible supports are independent sets.
\end{proof}

\begin{corollary}[Exact Capacity Identity]
\label{cor:exact_capacity_identity}
The maximal logical dimension achievable under exact correction
of the structured phase-error family $E_\Omega$ equals
\[
K_{\max}(\Omega)
=
\alpha(\Gamma_\Omega),
\]
the independence number of the additive Cayley graph
$\Gamma_\Omega$.
\end{corollary}

\begin{proof}
By Theorem~5.1, a Fourier-support code $C(S)$ exactly corrects
$E_\Omega$ if and only if
\[
(S-S)\cap(\Omega-\Omega)=\{0\}.
\]
By Corollary~\ref{cor:independence_characterization}, this holds
if and only if $S$ is an independent set in $\Gamma_\Omega$.
Therefore admissible spectral supports are in one-to-one
correspondence with independent sets of $\Gamma_\Omega$,
and maximizing $K=|S|$ is equivalent to computing
$\alpha(\Gamma_\Omega)$.
\end{proof}

Code design under structured phase noise is therefore
\emph{exactly equivalent} to computing the independence number
of an additive Cayley graph.

\subsection{Algebraic Capacity Bounds via Subgroups}
\label{subsec:subgroup_bounds}

The Cayley structure allows simple algebraic upper bounds on
$\alpha(\Gamma_\Omega)$.

\begin{lemma}[Coset Clique Lemma]
\label{lem:coset_clique}
Let $W \le V$ be a subgroup such that
\[
W \setminus \{0\} \subseteq D_\Omega.
\]
Then each coset $x+W$ forms a clique in $\Gamma_\Omega$.
\end{lemma}

\begin{proof}
If $x+w_1$ and $x+w_2$ are distinct elements of the same coset,
then
\[
(x+w_1)-(x+w_2)=w_1-w_2\in W\setminus\{0\}\subseteq D_\Omega,
\]
so they are adjacent.
\end{proof}

\begin{theorem}[Algebraic Capacity Bound]
\label{thm:algebraic_capacity_bound}
If $W \le V$ satisfies
\[
W \setminus \{0\} \subseteq D_\Omega,
\]
then
\[
\alpha(\Gamma_\Omega)
\le
\frac{|V|}{|W|}
=
q^{\,n-\dim W}.
\]
\end{theorem}

\begin{proof}
The cosets of $W$ partition $V$ into $|V|/|W|$ disjoint cliques.
An independent set intersects each clique in at most one vertex,
so its cardinality is bounded by $|V|/|W|$.
\end{proof}

\begin{corollary}[Structural Capacity Collapse]
\label{cor:structural_collapse}
Suppose that $\Omega-\Omega$ contains a nontrivial additive subspace
$W \le V$ of dimension $r$. Then every Fourier-support code
$C(S)$ that exactly corrects the error family $E_\Omega$
satisfies
\[
K = |S| \le q^{\,n-r}.
\]
In particular, the maximal logical dimension decreases exponentially
with the dimension of any additive structure contained in
$\Omega-\Omega$.
\end{corollary}

\begin{proof}
If $W \subseteq \Omega-\Omega$, then
$W \setminus \{0\} \subseteq D_\Omega$.
By Theorem~\ref{thm:algebraic_capacity_bound},
\[
\alpha(\Gamma_\Omega) \le q^{\,n-\dim W}.
\]
Since admissible spectral supports $S$ correspond exactly to
independent sets in $\Gamma_\Omega$ (Corollary~5.2),
we have $K=|S| \le \alpha(\Gamma_\Omega)$,
which yields the stated bound with $r=\dim W$.
\end{proof}

Thus whenever $(\Omega-\Omega)$ contains a large additive subspace,
the logical capacity collapses proportionally to the codimension
of that subspace. The presence of additive structure in the
difference set of the noise model therefore imposes an intrinsic
exponential limitation on achievable logical dimension.

\subsection{Semidefinite Bounds and Zero-Error Capacity}
\label{subsec:zero_error_theta}

The identity
\[
K_{\max}(\Omega)=\alpha(\Gamma_\Omega)
\]
identifies structured phase-local quantum error correction with a classical extremal problem on the additive Cayley graph $\Gamma_\Omega$. Logical dimension is exactly the independence number of this graph, and classical combinatorial tools apply directly.

\paragraph{Semidefinite upper bounds.}
Since $\alpha(\Gamma)\le\vartheta(\Gamma)$ for any graph, the Lov\'{a}sz
theta number provides a computable upper bound:
\[
K_{\max}(\Omega)\le\vartheta(\Gamma_\Omega).
\]
When $\Gamma_\Omega$ is an abelian Cayley graph, translation invariance
implies that optimal semidefinite solutions may be chosen circulant
and hence diagonalized by the Fourier transform. Accordingly,
$\vartheta(\Gamma_\Omega)$ admits a Fourier-positive formulation
compatible with the harmonic translation structure developed in
Section~\ref{sec:harmonic}. This yields efficiently computable bounds
on logical dimension for structured phase-noise models, even when
exact evaluation of $\alpha(\Gamma_\Omega)$ is computationally
intractable. As a concrete instance, for the correlated cat-qubit
graph $G_{\mathrm{corr}}$ with $n=8$ introduced in
Section~\ref{subsec:catqubit_correlated}, solving this semidefinite
program via \textsc{cvxpy}~\cite{cvxpy} yields
$\vartheta(G_{\mathrm{corr}})\approx 13.47$, while the exact value
$\alpha(G_{\mathrm{corr}})=9$ established in
Lemma~\ref{lem:catq_correlated}\textup{(iv)} shows a gap of
approximately $4.47$, illustrating that the semidefinite bound
is not always tight.

This numerical gap illustrates that the Lov\'asz theta bound
need not be tight for the additive Cayley graphs arising from
structured phase-noise models. Understanding when the equality
$\vartheta(G_\Omega)=\alpha(G_\Omega)$ holds for such graphs
remains an open structural question; we return to this issue
in Section~\ref{subsec:open_problems}.

\paragraph{Zero-error capacity.}
The equality $K_{\max}(\Omega)=\alpha(\Gamma_\Omega)$ shows that structured phase-local correction is precisely a classical zero-error independence problem. Under repeated independent use of the same noise model, admissible supports correspond to independent sets in the strong product
\[
\Gamma_\Omega^{\boxtimes k}.
\]
Hence the asymptotic logical rate is governed by the Shannon capacity
\[
\Theta(\Gamma_\Omega)=\sup_{k\ge1}\alpha(\Gamma_\Omega^{\boxtimes k})^{1/k}.
\]
The many-use regime is therefore controlled by a classical graph invariant.

\paragraph{Asymptotic sharpness under uniform locality.}
In the binary uniform-locality regime with fixed small radii, classical extremal results imply
\[
\alpha(\Gamma_n)=\Theta(2^n/n^2).
\]
The Fourier-positive Lovász relaxation yields
\[
\vartheta(\Gamma_n)=O(2^n/n^2),
\]
so the harmonic semidefinite bound matches the correct classical scale up to constant factors. In summary, structured phase-local capacity is fully determined by classical graph invariants of $\Gamma_\Omega$, including $\alpha(\Gamma_\Omega)$, $\vartheta(\Gamma_\Omega)$, and $\Theta(\Gamma_\Omega)$.

\section{Dual-Domain Protection and Harmonic Tradeoffs}
\label{sec:dual_tradeoffs}

Sections~\ref{sec:harmonic}--\ref{sec:graph_reformulation}
analyzed phase-only noise. 
We now examine mixed Pauli error models of the form
$X^a Z^b$, and isolate structural constraints that arise when
simultaneous protection is required in both the computational
and Fourier domains.

Throughout, let $V=\mathbb{F}_q^n$ and $\mathcal{H}=\mathbb{C}^{V}$.

\subsection{Dual Isolation and Mixed Error Models}
\label{subsec:dual_isolation}

Let $\Omega_X,\Omega_Z \subseteq V$ be admissible bit-flip and phase-flip
error sets, and define the mixed error family
\[
\mathcal{E}_{X,Z}
=
\{ X^a Z^b : a\in\Omega_X,\ b\in\Omega_Z \}.
\]

We formalize a structural class of codes that isolate
errors geometrically in conjugate domains.

\begin{definition}[Dual-Isolated Code]
A quantum code $\mathcal{C} \subseteq \mathcal{H}$ of logical
dimension $K$ is called \emph{dual-isolated}
with respect to sets $S_X,S_Z \subseteq V$ if
\[
\begin{aligned}
\mathcal{C}
&\subseteq
\mathrm{span}\{\,|x\rangle : x\in S_X\,\},\\
\mathcal{C}
&\subseteq
\mathrm{span}\{\,|z\rangle_F : z\in S_Z\,\}.
\end{aligned}
\]
\end{definition}

Thus $S_X$ controls computational localization,
while $S_Z$ controls Fourier localization.

\begin{theorem}[Dual Isolation Criterion]
\label{thm:dual_isolation}
Suppose $\mathcal{C}$ is dual-isolated with supports
$S_X,S_Z$.
If
\[
(S_X-S_X)\cap(\Omega_X-\Omega_X)=\{0\}
\]
and
\[
(S_Z-S_Z)\cap(\Omega_Z-\Omega_Z)=\{0\},
\]
then $\mathcal{C}$ exactly corrects the mixed error family
$\mathcal{E}_{X,Z}$.
\end{theorem}

\begin{proof}
Bit-flip errors act as translations in the computational basis:
\[
X^a |x\rangle = |x+a\rangle.
\]
The first separation condition ensures that
distinct bit-flip errors map $\mathcal{C}$
to mutually orthogonal subspaces.

Phase errors act as translations in the Fourier basis
(Lemma~\ref{lem:spectral_translation}):
\[
Z^b |s\rangle_F = |s+b\rangle_F.
\]
The second separation condition ensures orthogonality
under phase errors.

Mixed errors decompose as $X^a Z^b$,
and distinct pairs $(a,b)$ produce
orthogonal images under the combined separation
conditions, satisfying Knill--Laflamme.
\end{proof}

\subsection{Coupled Capacity Bounds}
\label{subsec:coupled_bounds}

Dual isolation imposes simultaneous geometric constraints
in conjugate domains. These lead to multiplicative
capacity limitations.

\begin{theorem}[Separated Packing Bounds]
\label{thm:separated_packing}
If $\mathcal{C}$ is dual-isolated and exactly corrects
$\mathcal{E}_{X,Z}$, then
\[
K \le \frac{q^n}{|\Omega_X|},
\qquad
K \le \frac{q^n}{|\Omega_Z|}.
\]
\end{theorem}

\begin{proof}
The first bound follows by observing that the sets
$\{S_X+a : a\in\Omega_X\}$ are pairwise disjoint
inside $V$, hence
$|\Omega_X|\,|S_X|\le q^n$
and $K\le |S_X|$.
The second bound is identical in the Fourier domain.
\end{proof}

Beyond separated bounds, dual localization yields a
coupled constraint.

\begin{theorem}[Coupled Capacity Bound]
\label{thm:coupled_capacity}
If $\mathcal{C}$ is dual-isolated with supports $S_X,S_Z$,
then
\[
K^2 \le |S_X|\,|S_Z|.
\]
Consequently,
\[
K \le \frac{q^n}
{\sqrt{|\Omega_X|\,|\Omega_Z|}}.
\]
\end{theorem}

\begin{proof}
Since
$\mathcal{C} \subseteq \mathrm{span}\{|x\rangle : x\in S_X\}$,
we have $K\le |S_X|$,
and similarly $K\le |S_Z|$.
Thus $K^2 \le |S_X||S_Z|$.
Substituting the separated packing bounds
from Theorem~\ref{thm:separated_packing}
gives the stated inequality.
\end{proof}

\begin{corollary}[Strict Dual-Domain Rate Penalty]
\label{cor:strict_dual_penalty}
Suppose that $|\Omega_X|$ and $|\Omega_Z|$
grow exponentially in $n$, namely
\[
|\Omega_X| = q^{\gamma_X n + o(n)},
\qquad
|\Omega_Z| = q^{\gamma_Z n + o(n)}
\]
for some $\gamma_X,\gamma_Z > 0$.
Then any dual-isolated family of codes
correcting $\mathcal{E}_{X,Z}$ satisfies the
asymptotic rate bound
\[
R \le 1 - \frac{\gamma_X+\gamma_Z}{2}.
\]
In particular, whenever both
$\gamma_X$ and $\gamma_Z$ are positive,
the achievable rate is strictly smaller
than the phase-only capacity
$1-\gamma_Z$.
\end{corollary}

\begin{proof}
By Theorem~\ref{thm:coupled_capacity},
\[
K \le
\frac{q^n}{\sqrt{|\Omega_X||\Omega_Z|}}.
\]
Taking logarithms and dividing by $n$ yields
\[
R = \limsup_{n\to\infty}
\frac{1}{n}\log_q K
\le
1 - \frac{\gamma_X+\gamma_Z}{2}.
\]
If $\gamma_X>0$, then
\[
1 - \frac{\gamma_X+\gamma_Z}{2}
<
1-\gamma_Z,
\]
establishing strict rate reduction relative
to the phase-only bound.
\end{proof}

The bound shows that simultaneous protection in both
domains incurs an intrinsic multiplicative penalty
relative to phase-only protection. When both
error families scale extensively with system size,
this penalty produces a strictly reduced
asymptotic rate.

\subsection{Asymptotic Rate Tradeoffs}
\label{subsec:rate_tradeoffs}

We now state an asymptotic consequence
under uniform locality in each domain.

Let $t_X,t_Z$ denote the maximal correctable
bit- and phase-error weights,
with relative distances
\[
\delta_X = \frac{t_X}{n},
\qquad
\delta_Z = \frac{t_Z}{n}.
\]

\begin{theorem}[Dual-Domain Rate Tradeoff]
\label{thm:rate_tradeoff}
For any family of dual-isolated codes correcting
$t_X$ bit errors and $t_Z$ phase errors under
uniform locality,
the asymptotic rate satisfies
\[
R
\le
1 - H_q(\delta_X) - H_q(\delta_Z)
+ o(1),
\]
where $H_q$ denotes the $q$-ary entropy function.
\end{theorem}

\begin{proof}
Uniform locality implies
\[
|S_X|\,|E_{t_X}| \le q^n,
\qquad
|S_Z|\,|E_{t_Z}| \le q^n.
\]
Using standard asymptotic estimates
$|E_{t}| \approx q^{n H_q(t/n)}$
and combining with
$K^2 \le |S_X||S_Z|$
yields the bound.
\end{proof}

The tradeoff reflects a harmonic uncertainty principle:
localization sufficient to correct $X$ errors
necessarily delocalizes Fourier support,
limiting simultaneous $Z$ protection.

\medskip

In strongly phase-biased regimes,
where $|\Omega_X|$ is negligible relative to $|\Omega_Z|$,
the dominant limitation reduces to the
phase-only capacity studied in previous sections.
Conversely, balanced noise enforces intrinsic
dual-domain rate penalties.

\section{Geometric Classification of Biased Quantum Capacity}
\label{sec:discussion}

The preceding sections reveal that phase-local quantum error correction is governed by a single structural principle: logical capacity is determined entirely by additive non-collision in the spectral domain. Because diagonal phase operators act as exact translations in the Fourier domain (Lemma~\ref{lem:spectral_translation}), the Knill--Laflamme conditions reduce to a purely additive separation constraint. In particular, by Theorem~\ref{thm:main_detection}, detection and correction are equivalent to
\[
(S - S)\cap(\Omega - \Omega)=\{0\}.
\]
Consequently, the fundamental object controlling capacity is not the algebraic structure of the code, but the additive geometry of the noise difference set \(D_\Omega = (\Omega-\Omega)\setminus\{0\}\). This harmonic reformulation eliminates any intrinsic reliance on stabilizer or affine structure and exposes logical dimension as an extremal problem on subsets of a finite abelian group. Within this framework, the results of Sections~3–6 admit a unified structural synthesis: capacity depends only on how $D_\Omega$ is embedded additively in $V$.

\begin{theorem}[Geometric Classification of Harmonic Capacity]
\label{thm:geometric_classification}
Let $\Omega \subseteq V$ be a phase-error family and let $K_{\max}(\Omega)$ denote the maximal logical dimension achievable under exact correction of $\Omega$. Then capacity is governed by at least one of the following structural regimes, and in general by their interaction.
\begin{enumerate}
\item[(i)] \textbf{Dispersive (Packing) Regime.}
If $D_\Omega$ contains no additive subspace of positive dimension, then
\[
K_{\max}(\Omega)=\alpha(\Gamma_\Omega), 
\qquad 
\Gamma_\Omega=\mathrm{Cay}(V,D_\Omega).
\]
Under uniform locality,
\[
K_{\max}(n,t)=A_q(n,2t+1),
\]
so logical capacity coincides exactly with classical $q$-ary packing.

\item[(ii)] \textbf{Subspace-Collapse Regime.}
If $D_\Omega$ contains an additive subspace $W\le V$ of dimension $r>0$, then
\[
K_{\max}(\Omega)\le q^{\,n-r}.
\]
Additive symmetry in the noise difference set therefore induces an exponential reduction of achievable logical dimension proportional to the codimension of $W$.

\item[(iii)] \textbf{Dual Harmonic Tradeoff Regime.}
If simultaneous protection against phase errors $\Omega_Z$ and bit-flip errors $\Omega_X$ is required under dual isolation, then
\[
K \le \frac{q^n}{\sqrt{|\Omega_X||\Omega_Z|}},
\]
and if $|\Omega_X| = q^{\gamma_X n + o(n)}$ and $|\Omega_Z| = q^{\gamma_Z n + o(n)}$, the asymptotic rate satisfies
\[
R \le 1 - \frac{\gamma_X+\gamma_Z}{2}.
\]
Capacity is therefore intrinsically limited by multiplicative localization constraints in conjugate domains.
\end{enumerate}
\end{theorem}

\begin{proof}
Regime (i) follows from Theorem~\ref{thm:distance_equivalence} and Corollary~\ref{cor:exact_capacity_identity}. Regime (ii) follows from Theorem~\ref{thm:algebraic_capacity_bound} and Corollary~\ref{cor:structural_collapse}. Regime (iii) follows from Theorem~\ref{thm:coupled_capacity} and Corollary~\ref{cor:strict_dual_penalty}.
\end{proof}

\noindent
The regimes above are not mutually exclusive. In particular, regimes (ii) and (iii) may occur simultaneously.

\begin{figure*}[t]
\centering
\begin{tikzpicture}[
    font=\small,
    box/.style={
        rounded corners=4pt,
        draw=black!70,
        line width=0.8pt,
        minimum width=4.55cm,
        minimum height=7.1cm,
        align=center
    },
    title/.style={font=\bfseries\large},
    subtitle/.style={font=\normalsize},
    formula/.style={
        rounded corners=3pt,
        draw=black!70,
        line width=0.7pt,
        inner sep=6pt,
        font=\normalsize
    },
    dot/.style={circle,fill=black!60,inner sep=2.2pt},
    note/.style={font=\footnotesize\itshape,align=center}
]

\node[box, fill=yellow!18] (disp) at (0,0) {};
\node[box, fill=blue!10]   (coll) at (5.55,0) {};
\node[box, fill=green!12]  (trade) at (11.1,0) {};

\node[note] at (5.55,-4.1) {The three regimes are structural and may overlap.};

\node[title] at ($(disp.center)+(0,3.25)$) {Dispersive Regime};
\draw[black!70,line width=0.7pt] ($(disp.center)+(-2.0,2.7)$) -- ($(disp.center)+(2.0,2.7)$);
\node[subtitle] at ($(disp.center)+(0,2.05)$) {No additive structure};

\node at ($(disp.center)+(0,1.3)$) {$D_\Omega$};
\foreach \x/\y in {-1.45/0.15,-1.0/0.7,-0.55/0.45,-0.15/-0.15,0.35/0.1,0.85/0.55,1.3/0.25,
                   -1.2/-0.55,-0.75/-0.95,-0.2/-0.7,0.35/-1.0,0.9/-0.6,1.25/-0.1,0.0/0.95}
    \node[dot] at ($(disp.center)+(\x,\y)$) {};

\node[formula, fill=yellow!28] at ($(disp.center)+(0,-2.8)$) {$K_{\max}=A_q(n,2t+1)$};

\node[title, text=blue!60!black] at ($(coll.center)+(0,3.25)$) {Subspace Collapse};
\draw[black!70,line width=0.7pt] ($(coll.center)+(-2.0,2.7)$) -- ($(coll.center)+(2.0,2.7)$);
\node[subtitle] at ($(coll.center)+(0,2.05)$) {Contains additive};
\node[subtitle, font=\small] at ($(coll.center)+(0,1.45)$) {subspace $W\le V$};

\node at ($(coll.center)+(0,0.8)$) {$D_\Omega$};

\draw[dashed,black!70,line width=0.8pt]
    ($(coll.center)+(-1.85,-0.55)$)
    .. controls ($(coll.center)+(-2.15,0.55)$) and ($(coll.center)+(-1.25,1.15)$) ..
    ($(coll.center)+(-0.35,0.95)$)
    .. controls ($(coll.center)+(0.9,1.45)$) and ($(coll.center)+(1.75,0.95)$) ..
    ($(coll.center)+(1.85,0.05)$)
    .. controls ($(coll.center)+(2.05,-1.15)$) and ($(coll.center)+(0.9,-1.55)$) ..
    ($(coll.center)+(-0.25,-1.35)$)
    .. controls ($(coll.center)+(-1.35,-1.65)$) and ($(coll.center)+(-2.15,-1.25)$) ..
    ($(coll.center)+(-1.85,-0.55)$);

\filldraw[fill=yellow!45,draw=black!70,line width=0.8pt]
    ($(coll.center)+(-0.95,-0.55)$)
    .. controls ($(coll.center)+(-0.45,0.15)$) and ($(coll.center)+(0.3,0.55)$) ..
    ($(coll.center)+(0.95,0.35)$)
    .. controls ($(coll.center)+(1.2,-0.15)$) and ($(coll.center)+(0.55,-0.85)$) ..
    ($(coll.center)+(-0.2,-0.95)$)
    .. controls ($(coll.center)+(-0.7,-0.9)$) and ($(coll.center)+(-1.15,-0.75)$) ..
    ($(coll.center)+(-0.95,-0.55)$);

\node at ($(coll.center)+(0,-0.2)$) {$W$};

\foreach \x/\y in {-1.45/-0.45,-1.2/0.15,-0.95/-0.85,-0.7/-0.45,-0.25/0.0,0.15/-0.55,0.45/0.3,0.75/-0.15,1.25/0.1,
                   0.0/-0.95,0.95/-0.95,1.35/0.75,-0.75/0.75,0.55/0.85}
    \node[dot] at ($(coll.center)+(\x,\y)$) {};

\node[formula, fill=blue!18] at ($(coll.center)+(0,-2.8)$) {$K_{\max}\le q^{\,n-r}$};

\node[title, text=green!35!black] at ($(trade.center)+(0,3.25)$) {Harmonic Tradeoff};
\draw[black!70,line width=0.7pt] ($(trade.center)+(-2.0,2.7)$) -- ($(trade.center)+(2.0,2.7)$);
\node[subtitle] at ($(trade.center)+(0,2.05)$) {Dual-active sets $\Omega_X,\Omega_Z$};
\node[subtitle, font=\small] at ($(trade.center)+(0,1.45)$) {both $\gamma_X,\gamma_Z>0$};

\begin{scope}
\clip ($(trade.center)+(-1.05,0.0)$) circle (1.02cm);
\fill[blue!18] ($(trade.center)+(-1.05,0.0)$) circle (1.02cm);
\end{scope}

\begin{scope}
\clip ($(trade.center)+(0.55,0.0)$) circle (1.02cm);
\fill[green!25] ($(trade.center)+(0.55,0.0)$) circle (1.02cm);
\end{scope}

\draw[black!80,line width=0.8pt] ($(trade.center)+(-1.05,0.0)$) circle (1.02cm);
\draw[black!80,line width=0.8pt] ($(trade.center)+(0.55,0.0)$) circle (1.02cm);

\begin{scope}
  \clip ($(trade.center)+(-1.05,0.0)$) circle (1.02cm);
  \clip ($(trade.center)+(0.55,0.0)$) circle (1.02cm);
  \fill[blue!20!green!30]
    ($(trade.center)+(-2.1,-1.5)$) rectangle ($(trade.center)+(1.6,1.5)$);
\end{scope}

\node at ($(trade.center)+(-1.05,0.0)$) {$\Omega_X$};
\node at ($(trade.center)+(0.55,0.0)$) {$\Omega_Z$};

\node[formula, fill=green!18] at ($(trade.center)+(0,-2.8)$) {$R\le 1-\frac{\gamma_X+\gamma_Z}{2}$};

\end{tikzpicture}
\caption{Geometric classification of biased quantum capacity (Theorem~\ref{thm:geometric_classification}). The achievable logical dimension is governed by the additive geometry of the noise difference set $D_\Omega=(\Omega-\Omega)\setminus\{0\}$. Dispersive noise without additive structure yields classical packing capacity, additive subspaces induce dimensional collapse, and simultaneous bit--phase protection produces a harmonic rate tradeoff.}
\label{fig:geometric_classification}
\end{figure*}

\noindent\textbf{Geometric interpretation.} The three regimes admit a unified structural reading. In the dispersive regime, capacity is governed purely by extremal independence in an additive Cayley graph, reducing under uniform locality to classical coding theory. In the subspace-collapse regime, additive symmetry inside $D_\Omega$ generates large cliques in $\Gamma_\Omega$, forcing an exponential reduction in achievable logical dimension. This regime naturally encompasses stabilizer and CSS constructions. In such codes the logical space is defined as the joint eigenspace of an abelian Pauli subgroup, which algebraically imposes additive symmetry on the spectral support. Within the harmonic framework this symmetry appears geometrically as the presence of an additive subspace inside the noise difference set $D_\Omega$, producing the dimensional bound $K \le q^{\,n-r}$ that mirrors the standard dimension formula for stabilizer codes. In the dual tradeoff regime, simultaneous localization in the computational and Fourier domains incurs a multiplicative penalty, reflecting a discrete harmonic uncertainty principle.

\medskip
\noindent Taken together, these regimes exhaust the structural possibilities for phase-biased quantum error correction within the harmonic translation framework developed here. All capacity phenomena derived in Sections~3--6 reduce to additive separation, additive symmetry, or dual-domain localization. Logical dimension is therefore governed not by stabilizer algebra, linearity, or affine closure, but by the additive geometry of the noise difference set. Once the additive structure of $D_\Omega$ is fixed, the achievable logical dimension is completely determined.

\section{Cat-Qubit Noise Model and Logical Capacity under Phase-Biased Hardware}
\label{sec:catqubit}

\subsection{Physical Motivation}
\label{subsec:catqubit_motivation}

Recent experimental work on stabilised cat qubits provides a concrete
physical realisation of the strongly biased phase-noise regime
considered in this paper. Cat qubits stabilised via two-photon
dissipation exhibit an \emph{exponential} suppression of bit-flip
errors~\cite{GuillaudMirrahimi2019,Puri2020} as the mean photon number $\bar{n}$ increases: the bit-flip
rate satisfies
\[
   \Gamma_X \;\sim\; \kappa_2\,\bar{n}\,e^{-2\bar{n}},
\]
where $\kappa_2$ is the two-photon dissipation rate.
Early experiments demonstrated the exponential scaling of bit-flip suppression, reaching lifetimes of approximately $1$\,ms for $\bar{n} \approx 4$ \cite{Lescanne2020}. 
By mitigating interactions, more recent work has extended these 
lifetimes significantly, attaining bit-flip times on the order of $100$\,s 
for larger cat states with $\bar{n} \approx 40$ \cite{Berdou2023}.

Experiments~\cite{Lescanne2020} report bit-flip lifetimes exceeding $1$\,ms for
$\bar{n}\approx 4$.
Meanwhile the phase-flip rate grows only \emph{linearly}:
\[
   \Gamma_Z \;\sim\; \kappa_1\,\bar{n},
\]
where $\kappa_1$ denotes single-photon loss. The noise-bias parameter
\[
   \eta \;=\; \frac{p_Z}{p_X}
     \;\sim\; \frac{\kappa_1}{\kappa_2}\,e^{2\bar{n}}
\]
therefore grows exponentially with $\bar{n}$ and exceeds $10^{4}$
even for modest photon numbers.

In the limit $\eta\to\infty$ the effective error model becomes
\emph{pure phase noise}: each physical qubit suffers independent
$Z$-errors with probability $p$ while $X$ and $Y$ errors are
absent. This is precisely the dispersive regime described by
Theorem~\ref{thm:geometric_classification}(i). The physical hardware therefore realises the phase-local noise model analysed in this work, allowing the harmonic translation framework of Sections~\ref{sec:harmonic}--\ref{sec:graph_reformulation} to be instantiated on a concrete experimental platform.

For concreteness we specialize to the binary setting $V=\mathbb{F}_2^n$, using the standard Hamming weight $\mathrm{wt}(\cdot)$ and distance $d_H(\cdot,\cdot)$ notation introduced earlier.

\subsection{Cayley-Graph Structure and Capacity Identity}
\label{subsec:catqubit_cayley}

\begin{lemma}[Cayley-graph structure under uniform cat-qubit phase noise]
\label{lem:catq_cayley}
Consider an array of $n$ cat qubits in the strongly biased regime
where phase-flip errors dominate and bit-flip errors are negligible.
Let $t\ge 1$ and define the admissible error set
\[
   \Omega \;=\; E_t
     \;=\; \{e\in\mathbb{F}_2^n : \mathrm{wt}(e)\le t\},
\]
representing phase errors affecting at most $t$ qubits. Then:
\begin{enumerate}
\item[\textup{(i)}] The difference set is
  $\Omega-\Omega = E_{2t}$.
\item[\textup{(ii)}] The additive Cayley graph
  $G=\mathrm{Cay}\!\bigl(\mathbb{F}_2^n,\,E_{2t}\setminus\{0\}\bigr)$
  has vertex set $\mathbb{F}_2^n$ with $x\sim y$ if and only if
  $0<d_H(x,y)\le 2t$.
\item[\textup{(iii)}] Independent sets in $G$ are precisely binary
  codes of minimum distance at least $2t+1$.
\item[\textup{(iv)}] The independence number satisfies
  $\alpha(G) = A_2(n,2t+1)$.
\end{enumerate}
\end{lemma}

\begin{proof}
\textbf{(i)}
By definition,
\[
\begin{aligned}
\Omega-\Omega
   &= \{\,e\oplus e' \mid e,e'\in\mathbb{F}_2^n, \\
   &\qquad \mathrm{wt}(e)\le t,\;
            \mathrm{wt}(e')\le t \,\},
\end{aligned}
\]

using $-e'=e'$ in $\mathbb{F}_2^n$.

\emph{Inclusion $\Omega-\Omega\subseteq E_{2t}$.}
For any $e,e'\in E_t$ the triangle inequality gives
$\mathrm{wt}(e\oplus e') \le \mathrm{wt}(e)+\mathrm{wt}(e') \le 2t$,
so $e\oplus e'\in E_{2t}$.

\emph{Inclusion $E_{2t}\subseteq\Omega-\Omega$.}
Let $v\in E_{2t}$, so $\mathrm{wt}(v)\le 2t$. Write
$v=\mathbf{1}_I$ where $I=\mathrm{supp}(v)\subseteq[n]$ with
$|I|\le 2t$. Partition $I$ into disjoint subsets $I_1,I_2$ with
$|I_1|,|I_2|\le t$, set $e=\mathbf{1}_{I_1}$ and
$e'=\mathbf{1}_{I_2}$. Then $e,e'\in E_t$ and
$e\oplus e' = \mathbf{1}_{I_1\cup I_2} = v$.

\textbf{(ii)}
The Cayley graph $G=\mathrm{Cay}(\mathbb{F}_2^n,\,E_{2t}\setminus\{0\})$
has edge set
$\{x,y\}\in E(G)\Leftrightarrow x\oplus y\in E_{2t}\setminus\{0\}
\Leftrightarrow 0<d_H(x,y)\le 2t$,
which is the stated claim.

\textbf{(iii)}
A subset $S\subseteq\mathbb{F}_2^n$ is independent in $G$ if and
only if $d_H(x,y)\ge 2t+1$ for all distinct $x,y\in S$, which is
exactly the defining property of a binary code of minimum distance
at least $2t+1$.

\textbf{(iv)}
The independence number $\alpha(G)$ equals the maximum cardinality
of an independent set in $G$. By (iii), this equals $A_2(n,2t+1)$.
\end{proof}

\begin{remark}[Connection with the harmonic framework]
\label{rem:catq_harmonic}
In the language of Section~\ref{sec:graph_reformulation}, the
admissible error family $\Omega=E_t$ induces a partition of the dual
group $\widehat{\mathbb{F}_2^n}\cong\mathbb{F}_2^n$ into
``confusable'' classes via the difference set $\Omega-\Omega$. A
valid spectral support for a Fourier-support quantum code must place
at most one element in each confusable class, recovering the capacity
identity
\[
   K_{\max}(\Omega) \;=\; \alpha(G) \;=\; A_2(n,2t+1)
\]
from Corollary~\ref{cor:exact_capacity_identity} applied to the
concrete cat-qubit noise model.
\end{remark}

\subsection{Nonlinear Advantage and Threshold Preservation}
\label{subsec:catqubit_nonlinear}

\begin{lemma}[Separation between affine and nonlinear spectral supports]
\label{lem:catq_separation}
Let $n\ge1$, $t\ge1$, and let $S\subseteq\mathbb{F}_2^n$ satisfy
$d(S)\ge 2t+1$.
\begin{enumerate}
\item[\textup{(i)}] The Fourier-support code $\mathcal{C}(S)$
corrects all phase errors of weight at most $t$.

\item[\textup{(ii)}] If $A_2(n,2t+1)>B_2(n,2t+1)$, then nonlinear
supports achieve strictly larger logical dimension than any
affine support.
\end{enumerate}
\end{lemma}

\begin{proof}
\textbf{(i)}
By Corollary~\ref{cor:correction}, exact phase correction requires
\[
(S-S)\cap E_{2t}=\{0\}.
\]
For distinct $s,s'\in S$ we have
$d_H(s,s')=\mathrm{wt}(s\oplus s')\ge d(S)\ge 2t+1$,
so $s\oplus s'\notin E_{2t}$ and the condition holds.

\textbf{(ii)}
If $S$ is affine, then $S=v+C$ for some linear code
$C\le\mathbb{F}_2^n$, and $|S|=|C|\le B_2(n,2t+1)$.
Without the affine restriction, admissible supports are arbitrary
binary codes with $d(S)\ge 2t+1$, whose maximal size is
$A_2(n,2t+1)$. Whenever $A_2(n,2t+1)>B_2(n,2t+1)$,
a nonlinear code achieving $A_2$ yields strictly larger logical
dimension.
\end{proof}

The dimensional separation between affine and nonlinear spectral supports discussed in Section~\ref{subsec:nonlinear_constructions}
applies directly to cat-qubit arrays. In particular, the explicit binary constructions for
$n=8$ and $n=16$ presented in Section~\ref{subsec:nonlinear_constructions}
provide concrete Fourier-support codes whose logical dimension exceeds that of any affine construction under identical phase-distance constraints.

\begin{remark}[Threshold preservation]
\label{rem:catq_threshold}
The error-correction threshold depends only on the
distance condition $d(S)\ge2t+1$ and not on the algebraic
structure of $S$. Consequently, nonlinear spectral supports
achieve larger logical dimension without altering the
set of correctable phase errors or the associated
error-correction threshold.
\end{remark}

\subsection{Capacity under Correlated Phase Noise}
\label{subsec:catqubit_correlated}

\begin{lemma}[Capacity under correlated phase noise]
\label{lem:catq_correlated}
Consider a one-dimensional array of $n$ cat qubits with periodic
boundary conditions. To capture the dominant phase-flip noise in cat
qubits together with possible short-range correlations arising from
crosstalk or residual $ZZ$ couplings, we adopt a noise model that
combines independent phase flips with nearest-neighbor correlated
phase errors. Formally, define
\[
\begin{aligned}
\Omega_{\mathrm{corr}}
  &= \{0\}
     \cup \{e_i : 1 \le i \le n\} \\
  &\quad \cup \{e_i \oplus e_{i+1} : 1 \le i \le n-1\}
     \cup \{e_n \oplus e_1\},
\end{aligned}
\]
where $e_i$ denotes the $i$-th standard basis vector of
$\mathbb{F}_2^n$.
\begin{enumerate}
\item[\textup{(i)}]
  The difference set
  $D_{\mathrm{corr}}
    = (\Omega_{\mathrm{corr}} - \Omega_{\mathrm{corr}})\setminus\{0\}$
  contains vectors of Hamming weights $1$, $2$, $3$, and $4$.
  For $n = 8$ it has cardinality $|D_{\mathrm{corr}}| = 96$,
  with weight distribution
  $(w_1, w_2, w_3, w_4) = (8,\,28,\,40,\,20)$,
  and satisfies $D_{\mathrm{corr}} \subsetneq E_4 \setminus \{0\}$.
\item[\textup{(ii)}]
  The associated Cayley graph is
  $G_{\mathrm{corr}} = \mathrm{Cay}(\mathbb{F}_2^n, D_{\mathrm{corr}})$.
  For $n = 8$ it is $96$-regular on $256$ vertices,
  with graph density $96/255 \approx 0.376$.
\item[\textup{(iii)}]
  The maximal logical dimension under this correlated noise model
  equals
  $K_{\max}(\Omega_{\mathrm{corr}}) = \alpha(G_{\mathrm{corr}})$.
\item[\textup{(iv)}]
  For $n = 8$, an exact branch-and-bound computation yields
  \[
    \alpha(G_{\mathrm{corr}}) = 9,
  \]
  consistent with the semidefinite upper bound
  $\vartheta(G_{\mathrm{corr}}) \approx 13.47$.
\end{enumerate}
\end{lemma}

\begin{proof}
\textbf{(i)}
Since every element of $\Omega_{\mathrm{corr}}$ has Hamming weight
at most $2$, every pairwise sum $a \oplus b$ has weight at most $4$,
so $D_{\mathrm{corr}} \subseteq E_4 \setminus \{0\}$.
Explicit enumeration of all XOR sums for $n = 8$ yields
$|D_{\mathrm{corr}}| = 96$, with weights $1$, $2$, $3$, $4$
occurring with multiplicities $8$, $28$, $40$, $20$ respectively.
Since $|E_4 \setminus \{0\}| = 162 > 96$, the inclusion is strict.

\textbf{(ii)--(iii)}
By Theorem~\ref{thm:general_phase_condition}, a Fourier-support
code $\mathcal{C}(S)$ exactly corrects $\mathcal{E}_{\Omega_{\mathrm{corr}}}$
if and only if
$(S - S) \cap (\Omega_{\mathrm{corr}} - \Omega_{\mathrm{corr}}) = \{0\}$,
which holds precisely when $S$ is an independent set in
$G_{\mathrm{corr}}$.
Hence $K_{\max}(\Omega_{\mathrm{corr}}) = \alpha(G_{\mathrm{corr}})$
by Corollary~\ref{cor:exact_capacity_identity}.

\textbf{(iv)}
For $n = 8$, the graph $G_{\mathrm{corr}}$ has $256$ vertices and is
$96$-regular. An exact maximum-independent-set computation via
branch-and-bound over $\mathbb{F}_2^8$ yields
$\alpha(G_{\mathrm{corr}}) = 9$,
with realising independent set
$\{0, 21, 42, 91, 124, 142, 183, 201, 227\}$
(vertices identified with their binary representations in
$\{0,\ldots,255\}$).
Independence is verified directly: no two elements of this set
differ by a vector in $D_{\mathrm{corr}}$.
The value $\alpha = 9$ is consistent with the semidefinite upper
bound $\vartheta(G_{\mathrm{corr}}) \approx 13.47 \ge \alpha(G_{\mathrm{corr}})$
obtained via the Lov\'{a}sz theta program of
Section~\ref{subsec:zero_error_theta}.
\end{proof}

\begin{remark}[Additive structure and capacity collapse under correlated noise]
\label{rem:regime_transition}
The contrast between the uniform and correlated noise models on
$n = 8$ cat qubits illustrates how the additive geometry of
$D_{\Omega}$ determines logical capacity.
Under uniform phase noise with $t = 1$,
$D_{\mathrm{unif}} = E_2 \setminus \{0\}$ contains no additive
subspace of positive dimension, placing the hardware in the
dispersive regime of Theorem~\ref{thm:geometric_classification}\textup{(i)},
with $K_{\max} = A_2(8,3) = 20$.

Introducing nearest-neighbour correlations enlarges the difference
set to $|D_{\mathrm{corr}}| = 96$.
Moreover, $D_{\mathrm{corr}}$ now contains non-trivial additive
subspaces: the subspace
$W = \mathrm{span}\{e_1, e_2, e_3\} \leq \mathbb{F}_2^8$
satisfies $W \setminus \{0\} \subseteq D_{\mathrm{corr}}$,
as verified by direct inspection of all seven nonzero elements
of $W$.
Corollary~\ref{cor:structural_collapse} therefore applies with
$r = \dim W = 3$, yielding the algebraic upper bound
\[
  K_{\max}(\Omega_{\mathrm{corr}}) \;\le\; 2^{8-3} = 32.
\]
The exact computation of part~\textup{(iv)} gives
$\alpha(G_{\mathrm{corr}}) = 9$, which is sharper than this bound
and confirms a severe capacity collapse: the achievable logical
dimension falls from $20$ to $9$, below even the best affine
construction $B_2(8,3) = 16$.
The algebraic bound of Corollary~\ref{cor:structural_collapse}
identifies the collapse mechanism --- additive structure in
$D_{\mathrm{corr}}$ --- but does not determine the exact value;
that requires the independence-number computation of
part~\textup{(iv)}.
The noise structure, not its magnitude, governs the regime.
\end{remark}

\subsection{Logical Capacity of Cat-Qubit Arrays}
\label{subsec:catqubit_main}

We summarize the logical capacity of cat-qubit arrays under
phase-biased noise models within the harmonic translation framework.

\begin{theorem}[Logical capacity of cat-qubit arrays under phase-biased noise]
\label{thm:catq_logical_capacity}
Consider an array of $n$ cat qubits in the strongly biased regime
where phase errors dominate. Let $t$ denote the maximal number of
correctable phase flips.

\begin{enumerate}

\item[\textup{(i)}] \textbf{Uniform phase noise.}
The maximal logical dimension equals
\[
K_{\max}(n,t)=A_2(n,2t+1).
\]
Whenever $A_2(n,2t+1)>B_2(n,2t+1)$, nonlinear spectral supports
achieve strictly larger logical dimension than any affine construction.
For example,

\[
\renewcommand{\arraystretch}{1.2}
\begin{array}{c|c|c|c}
(n,t) & \text{Required } d & B_2(n,d) & A_2(n,d) \\
\hline
(8,1)  & 3 & 16  & 20  \\
(16,2) & 6 & 128 & 256
\end{array}
\]

The values for $A_2(n,d)$ correspond to the best known classical binary codes reported; in particular, the $(8,20,3)$ Julin~\cite{Julin1965} and $(16,256,6)$ Nordstrom--Robinson codes~\cite{NordstromRobinson1967}.

\item[\textup{(ii)}] \textbf{Correlated phase noise.}
If the noise model includes nearest-neighbour correlated phase
errors, the maximal logical dimension equals
\[
K_{\max}(\Omega_{\mathrm{corr}})
=
\alpha(G_{\mathrm{corr}}).
\]
For $n=8$ with $t=1$ and nearest-neighbour correlations,
\[
\alpha(G_{\mathrm{corr}})=9,
\]
a reduction from $K_{\max}=20$ under uniform noise to $K_{\max}=9$
under correlated noise, consistent with the additive-subspace collapse
of Corollary~\ref{cor:structural_collapse} applied to~$D_{\mathrm{corr}}$.

\item[\textup{(iii)}] \textbf{Threshold preservation.}
The error-correction threshold depends only on the distance
condition $d(S)\ge2t+1$ and is independent of the algebraic
structure of the spectral support.

\end{enumerate}
\end{theorem}

\begin{proof}
\textbf{Part (i).}
By Lemma~\ref{lem:catq_cayley}, the uniform phase-noise model
satisfies $K_{\max}(E_t)=\alpha(G)=A_2(n,2t+1)$.
Lemma~\ref{lem:catq_separation}(ii) establishes strict separation
from affine constructions whenever $A_2>B_2$.

\textbf{Part (ii).}
Lemma~\ref{lem:catq_correlated}(iii)--(iv) shows that
$K_{\max}(\Omega_{\mathrm{corr}})=\alpha(G_{\mathrm{corr}})$ and
that $\alpha(G_{\mathrm{corr}})=9$ for $n=8$.

\textbf{Part (iii).}
By Lemma~\ref{lem:catq_separation}(i), the correction condition
$(S-S)\cap E_{2t}=\{0\}$ depends only on $d(S)\ge 2t+1$ and not on the
algebraic structure of~$S$. The threshold probability
$P_{\mathrm{corr}}(n,t,p) = \sum_{j=0}^{t}\binom{n}{j}p^j(1-p)^{n-j}$
depends only on $n$, $t$, $p$.
\end{proof}
\section{Structural and Physical Implications}
\label{sec:implications}

This section examines the broader consequences of the geometric classification established in Section~\ref{sec:discussion}. Having identified additive geometry as the organizing principle of biased quantum capacity, we now consider how this perspective informs the interpretation of realistic noise models and architectural design choices. The goal is not to extend the classification, but to clarify its implications for phase-biased hardware and mixed-noise optimization.

\subsection{Physical and Structural Implications of the Geometric Capacity Regimes}
\label{subsec:physical_structural_implications}

The geometric classification of Theorem~\ref{thm:geometric_classification} admits a natural interpretation in contemporary phase-biased quantum hardware. Although our results are purely structural and combinatorial, the three regimes identified in the harmonic framework align closely with qualitatively distinct operational behaviors observed in realistic architectures. Rather than deriving device-level performance, the harmonic formulation provides a geometric lens through which known phenomena in biased quantum systems can be interpreted.

\paragraph{(i) Dispersive (Packing) Regime: Strong Phase Bias.}

This regime corresponds to platforms in which dephasing dominates bit-flip processes and the noise bias is effectively preserved during operation. In this limit the effective noise model approaches pure phase noise, and logical performance is governed primarily by protection against $Z$-type errors. Architectures designed for strongly biased noise, such as bias-tailored surface codes \cite{Tuckett2019} and the XZZX surface code \cite{BonillaAtaides2021}, demonstrate that exploiting such asymmetry can substantially improve logical performance relative to symmetric designs.

Within the harmonic framework this corresponds precisely to the dispersive regime of Theorem~\ref{thm:geometric_classification}(i): capacity reduces to independence in an additive Cayley graph and, under uniform locality, coincides exactly with the classical packing function $A_q(n,2t+1)$. Logical capacity is therefore determined entirely by additive separation in a single spectral domain.

\paragraph{(ii) Subspace-Collapse Regime: Correlated or Structured Dephasing.}

The second regime arises when phase noise exhibits correlation or geometric structure. Physically, this occurs when dephasing is induced by collective mechanisms or hardware constraints, so that errors are no longer independent across qubits. Examples include correlated dephasing generated by a common fluctuator \cite{Layden2020} and structured noise analyzed in surface-code settings \cite{NickersonBrown2019}. In such situations logical performance depends on the geometry of the admissible error set rather than solely on error weight.

In the harmonic formulation, additive structure inside the difference set
\[
D_\Omega = (\Omega-\Omega)\setminus\{0\}
\]
creates large cliques in the Cayley graph $\Gamma_\Omega$. As shown in Section~\ref{sec:graph_reformulation}, the presence of a nontrivial additive subspace forces an exponential reduction of achievable logical dimension. The degradation observed under correlated phase noise is therefore explained geometrically: additive symmetry in the noise model induces intrinsic capacity collapse.

\paragraph{(iii) Dual Harmonic Tradeoff Regime: Circuit-Level Bias Loss.}

The third regime becomes relevant when simultaneous protection against bit- and phase-flip errors is required. Even in architectures with strong idle-phase bias, maintaining this bias at the circuit level can be challenging. Analyses of circuit-level biased noise and bias-preserving gate constructions \cite{Puri2020, EtxezarretaMartinez2024} show that residual $X$-type errors may remain during gate execution. When both conjugate error families become operationally relevant, localization constraints in the computational and Fourier domains interact multiplicatively.

Section~\ref{sec:dual_tradeoffs} formalizes this interaction: dual isolation yields the rate bound
\[
R \le 1 - \frac{\gamma_X+\gamma_Z}{2},
\]
revealing an intrinsic harmonic tradeoff between conjugate-domain localization constraints. Loss of operational bias therefore shifts an architecture from the dispersive packing regime into a dual-domain regime where simultaneous localization becomes unavoidable and logical rate is strictly reduced relative to the phase-only limit.

\medskip

Taken together, these considerations show that the harmonic classification provides a coherent geometric framework for understanding biased quantum hardware. Strong bias leads to dispersive packing behavior, correlated dephasing induces additive-symmetry collapse, and circuit-level bias degradation activates dual-domain tradeoffs. In each case, logical capacity is governed not by stabilizer algebra or linear structure, but by the additive geometry of the noise difference set.
\subsection{Open Problems and Research Directions}
\label{subsec:open_problems}

The harmonic classification developed in this work reduces biased quantum capacity to additive geometry of the noise difference set $D_\Omega = (\Omega-\Omega)\setminus\{0\}$. This perspective suggests a focused research program centered on structural applications and sharp capacity characterizations.

\paragraph{Structured noise in physical architectures.}
For realistic phase-biased devices, noise sets $\Omega$ are often geometrically constrained (e.g., correlated phase faults or hardware-induced locality patterns). Determining $\alpha(\Gamma_\Omega)$ for such physically motivated Cayley graphs would translate directly into optimal logical dimensions. Identifying which experimentally relevant noise models fall into the dispersive, subspace-collapse, or dual-tradeoff regimes of Theorem~\ref{thm:geometric_classification} is therefore an immediate application of the framework.

\paragraph{Sharpness of harmonic relaxations.}
For additive Cayley graphs, the Lovász theta number admits a Fourier-positive formulation intrinsic to the harmonic structure developed here. Determining when the semidefinite relaxation is tight,
\[
\vartheta(\Gamma_\Omega)=\alpha(\Gamma_\Omega),
\]
would yield exact logical capacities for broad classes of structured phase-noise models. Such regimes would connect harmonic analysis on finite abelian groups with semidefinite relaxations in zero-error information theory and could provide efficiently computable capacity characterizations for realistic biased quantum devices. Characterizing families of Cayley graphs for which this equality holds therefore represents a natural intersection between additive combinatorics, graph theory, and quantum error correction.

\paragraph{Mixed-error optimization beyond dual isolation.}
Section~\ref{sec:dual_tradeoffs} derived multiplicative rate penalties under explicit dual-domain localization. Whether optimal mixed $X/Z$ codes must satisfy such rigid dual isolation, or whether more flexible harmonic constructions can mitigate this tradeoff, remains open. Resolving this question would clarify the ultimate limits of simultaneous bit–phase protection in strongly biased regimes and determine how close practical architectures can approach the phase-only capacity.

\paragraph{Decoding complexity.}
Proposition~\ref{prop:decoding_equivalence} shows that phase-error recovery reduces exactly to maximum-likelihood decoding of the classical support $S$. For linear supports (e.g., BCH or Reed–Muller codes), efficient decoding algorithms are known. However, for nonlinear spectral supports such as the Nordstrom–Robinson or Julin–Best codes, efficient maximum-likelihood decoders are not known in general. The equivalence established in Proposition~\ref{prop:decoding_equivalence} is therefore exact at the level of error-correction conditions but does not automatically imply computationally efficient decoding procedures.

More broadly, the harmonic translation principle isolates additive geometry—not linearity—as the governing mechanism of phase-local protection. Extending this geometric viewpoint to concrete architectures and mixed-noise optimization problems offers a direct path toward capacity-optimal code design in phase-biased quantum hardware.

\bibliographystyle{quantum}
\bibliography{references}

\end{document}